\documentclass[conference]{IEEEtran}
\IEEEoverridecommandlockouts 

\usepackage{cite}
\usepackage{amsmath,amssymb,amsfonts}
\usepackage{algorithmic}
\usepackage{graphicx}
\usepackage{textcomp}
\usepackage{xcolor}
\usepackage{placeins}  
\usepackage{afterpage}
\usepackage{xurl}

\usepackage[normalem]{ulem}
\usepackage{booktabs}
\usepackage[ruled,vlined]{algorithm2e}
\def\BibTeX{{\rm B\kern-.05em{\sc i\kern-.025em b}\kern-.08em
    T\kern-.1667em\lower.7ex\hbox{E}\kern-.125emX}}

\usepackage{tikz}
\usetikzlibrary{
    automata,
    shapes,
    arrows,
    positioning
}\usepackage{amsfonts, amsmath, amssymb, amsthm}

\usepackage[dvipsnames]{xcolor}

\newtheorem{claim}{Claim}

\begin{document}

\bibliographystyle{IEEEtran}

\title{Tools for Reducing Service Time in Near-Term Quantum Networks
{\footnotesize }
\thanks{This work was funded by the European Union's Horizon Europe research and innovation programme under grant agreement No. 101102140 – QIA Phase 1. The authors also acknowledge funding from NWO VICI.}
}
\author{
\IEEEauthorblockN{
Jake Smith\textsuperscript{1,2,3},
Thomas R. Beauchamp\textsuperscript{1,2,3},
Scarlett Gauthier\textsuperscript{1,2,3},
Oumayma Bouchmal\textsuperscript{1,2,3},
and Stephanie Wehner\textsuperscript{1,2,3}
}
\IEEEauthorblockA{
\textsuperscript{1}QuTech, Delft University of Technology \\
\textsuperscript{2}Kavli Institute of Nanoscience, Delft University of Technology \\
\textsuperscript{3}Quantum Computer Science, Electrical Engineering, Mathematics and Computer Science, Delft University of Technology
\\
\textsuperscript{} jsssmith@tudelft.nl, t.r.beauchamp@tudelft.nl, s.s.gauthier@tudelft.nl, o.bouchmal@tudelft.nl, s.d.c.wehner@tudelft.nl
}
}

\maketitle
\begin{abstract}
Architectures have been proposed to control entanglement generation in multi-user quantum networks. To allow time for local operations and classical communication at end nodes, these architectures insert fixed separations between consecutive batches of entanglement generation attempts. This reduces network utilization when attempts fail, leaving the network idle during the scheduled separation. To address this limitation, we propose a novel method to reclaim this idle time by shortening the scheduled separation between attempts while respecting hardware constraints. The method uses an analytical execution model to optimize the separation and reduce the total network service time of an application. Evaluations within the Arqon architecture show network service time reductions of up to $42$~minutes ($7.6\%$) for single applications and $16$--$29$~minutes ($26$--$30\%$) per application when co-scheduled. The approach applies broadly to quantum network architectures that share hardware between entanglement generation and local operations, and the method can be used online by network schedulers.
\end{abstract}

\begin{IEEEkeywords}
Quantum Internet, entanglement scheduling, minimum separation.
\end{IEEEkeywords}

\section{Introduction}
Quantum networks generate shared entanglement between end nodes, enabling the execution of applications such as secure remote computation~\cite{broadbent2009universal}, improved clock synchronization~\cite{giovannetti2001quantum}, precision sensing~\cite{zaiser2016enhancing}, and others~\cite{wehner2018quantum, bova2021commercial}. Application execution at end nodes is organized into an \emph{application session}, defined by a target number of application instances to be completed. A session describes the application’s demand for entanglement over time, requiring the network to repeatedly generate entanglement in order to serve successive instances.

To support the execution of each application instance, the network generates end‑to‑end entanglement as \emph{entangled links}. The network computes a schedule that coordinates entanglement generation attempts along a path, organizing them into scheduled blocks of time during which attempts are retried until the required entanglement is produced or the block ends. The \emph{network service time} is the total time scheduled by the network for a session.

The end nodes at either end of the path are not under network control and must actively participate for an entangled link to be generated. Two distinct roles arise within a scheduled block. The network \emph{delivers} an entangled link when it successfully produces one spanning the interior of the path. The end nodes then \emph{consume} a delivered link by aligning their own generation attempts with the network’s schedule. After consumption, the end nodes perform local operations and classical
communication (LOCC), such as quantum gates,
to complete the execution of an application instance.

In demonstrated quantum networks, end node hardware participating in consumption also performs LOCC (e.g.~\cite{liu2026long, stolk2024metropolitan, pompili2022experimental, two_level_control, islam2025experimental, kapoor2025public}). 
As a result, consecutive consumptions must be separated by at least the time required for LOCC, an interval we refer to as the \emph{minimum separation}.
The network scheduler must therefore space consecutive deliveries by at least this interval, making schedule computation a constrained scheduling problem.

The impact of the minimum separation constraint depends on how the
minimum separation compares to the expected interval between
successive consumptions. On current hardware, these two timescales are
comparable. Demonstrated heralded entanglement generation rates are
$2.22~\mathrm{Hz}$ in trapped-ion systems~\cite{liu2026long} and
$0.022~\mathrm{Hz}$ in solid-state platforms~\cite{stolk2024metropolitan},
both over $10~\mathrm{km}$ links. Minimum separation times arise from
device-level operations such as quantum gates, cooling, and
readout~\cite{krutyanskiy2023telecom, canteri2024photon, hucul2015modular},
as well as from control-stack overhead including classical processing
and communication~\cite{delle2025operating, van2025qoala}. These
contributions span microseconds to seconds. At the trapped-ion rate,
successive consumptions occur every $450~\mathrm{ms}$ on average,
comparable to a representative minimum separation of approximately
$500~\mathrm{ms}$, making the separation sufficiently long to
constrain every consumption.

A conservative scheduling approach assumes that end nodes always require
the full minimum separation between consecutively scheduled blocks. Because
each block succeeds only probabilistically, many blocks fail to produce
an end-to-end entangled link. Nevertheless, the minimum separation is
scheduled after every block, so following a failure the end nodes remain
idle for the full interval despite having no LOCC to perform.
This behavior over-provisions the network service time and prolongs session execution time.

To reduce this idle time, it is natural to consider shortening the
scheduled separation below the minimum. Doing so, however, risks
scheduling a block while the end nodes are still performing LOCC. In this case, the block must be skipped or the current
application instance preempted, discarding in-progress computation.
Shorter separations therefore reduce network service time but increase
the probability that blocks are skipped. Enough skips will prevent the session from completing its target number of
instances. This tradeoff admits an optimal scheduled separation that
minimizes session execution time while still completing the target
instances. 

In this work, we address this challenge by developing analytical tools
that compute the optimal scheduled separation for a given application
session. Because the minimum separation constraint arises whenever the
end node hardware that participates in consumption also performs LOCC,
the proposed tools apply broadly to entanglement delivery scheduling and
are not tied to a specific network architecture. We evaluate the tools
within the Arqon~\cite{arqon_arxiv} suite of network control applications,
which accepts entanglement demands from end nodes and computes schedules
that respect the minimum separation constraint. Arqon commits to
delivering the entangled links required by a session with at least a
specified service probability, providing a concrete architecture in which
to quantify the benefits of an optimized scheduled separation. The
contributions of this work are as follows:
\begin{itemize}
    \item \textbf{Analytical optimization of scheduled separation}: We
develop a novel method to align the network's delivery schedule with
end node consumption by modeling session execution under a reduced
scheduled separation as a Markov chain over executed and skipped blocks. 
From this we
formalize a constrained optimization to determine the optimal
scheduled separation and derive a closed-form lower bound on the
expected execution speedup. The resulting tools are lightweight and
can be used by network schedulers at runtime without repeated
simulation.
    \item \textbf{Evaluation under realistic hardware parameters}: We
evaluate the proposed tools within the Arqon
architecture~\cite{arqon_arxiv} using state-of-the-art and projected
trapped-ion entanglement generation parameters. In the state-of-the-art regime, the network service time is reduced by
up to $42$ minutes, a $7.6\%$ reduction corresponding to a single-session
execution speedup of $6.1\%$. In the projected regime, the network service time is reduced by
up to $2.8$ hours, a $6.5\%$ reduction and a speedup of $5.1\%$. We identify the conditions
required for positive speedup and validate the analytical predictions
using Monte Carlo simulation, confirming that the requested service
probabilities are maintained across all tested configurations.
\item \textbf{Extension to multi-session scheduling}: We extend the analysis to multiple sessions co-scheduled under
round-robin scheduling, showing that this configuration provides
additional separation that mitigates or eliminates skip penalties. For two co-scheduled sessions, the network service time is reduced by
$16$--$29$~minutes per session, a $26$--$30\%$ reduction corresponding
to execution speedups of $25$--$27\%$, substantially exceeding the
single-session case while maintaining the required service probability.
\end{itemize}

\section{Related Work}
The minimum separation constraint poses an explicit scheduling
problem only when three conditions hold together. Entanglement must
be generated when requested rather than buffered in advance; each end
node must manage its own execution, as in quantum internet (QI)
applications; and the end node must share hardware between consumption
and LOCC, since dedicated hardware for each would let local operations
proceed without blocking the next consumption.

Demonstrated quantum networks operate in the generate-when-requested
(GWR) paradigm, generating entanglement in response to application
demand rather than in advance~\cite{liu2026long,
stolk2024metropolitan, pompili2022experimental, two_level_control,
islam2025experimental, kapoor2025public}. A pre-loaded network
instead generates and buffers entanglement independently of requests,
and applications draw from the buffer when ready to
consume~\cite{gu2023esdi, pirker2019quantum, pouryousef2023quantum}
Each consumption then proceeds when ready and respects the separation
by construction. Such architectures, however, require generation
rates and coherence times beyond current
capabilities~\cite{beauchamp2025modular} and have not been
demonstrated.

The constraint also depends on which entity schedules end node
execution. In distributed quantum computing, a central orchestrator
schedules entanglement deliveries, consumption, and local operations
jointly~\cite{two_level_control, islam2025experimental,
kapoor2025public}, again satisfying the constraint by construction. QI
applications, by contrast, partition the application into node-specific
programs scheduled independently by each end node. The network
schedule coordinates delivery but cannot guarantee the constraint is
respected. Even so, most scheduling work for GWR networks serving QI
applications terminates the model at the generation of an end-to-end
link, abstracting delivery and consumption into a single
event~\cite{shi2020concurrent, zhao2021redundant, pant2019routing,
li2021effective, cicconetti2021request, dai2020optimal, gu2023esdi,
vasantam2021stability}, so the local operations that follow
consumption on shared hardware never enter the optimization.

The constraint therefore has not been observed as a scheduling problem
in any demonstration. The networks that both generate on demand and
share end node hardware either orchestrate execution centrally, or
have not yet performed QI sessions with repeated application instances
in which it manifests. This is the regime in which near-term quantum
networks will operate and QI applications will be programmed, and
recent work has begun to develop end node and network architectures to
support it.

QNodeOS, implemented by Delle Donne~\textit{et~al.}~\cite{
delle2025operating}, is the first end node runtime enabling
platform-independent execution of QI applications written in high-level
software. To provide a compatible network architecture, Beauchamp~\textit{et~al.}~\cite{beauchamp2025modular} propose a centrally controlled architecture organized around the \emph{entanglement packet}: a specification of the entangled links required to execute a single application instance, generated within a scheduled interval referred to as a packet generation attempt (PGA). To respect the minimum separation constraint, consecutive PGAs are conservatively separated by the full minimum separation, assuming the worst case in which the packet is generated at the very end of the PGA. This is notably the first architecture to consider how application execution requirements must inform network scheduling.

While the network architecture requires that successive packets be separated by some minimum separation, it does not prescribe how packet deliveries are scheduled or with what reliability. To address this, Arqon~\cite{arqon_arxiv} builds on this architecture as the first suite of network control applications designed to reliably fulfill accepted demands by periodically computing a network schedule of PGAs. For each accepted demand, a service agreement specifies the number of packets required to complete all instances, and Arqon schedules sufficient PGAs to meet this requirement with the service probability before the expiry time. For each scheduled PGA, Arqon reserves the appropriate network resources on an end-to-end path between the two end nodes executing the application.

The Arqon authors further provide a dedicated simulation framework, \texttt{arqon-sim}~\cite{arqon_sim}, which simulates schedule execution at the level of individual PGAs: a PGA either succeeds or fails with the packet success probability, and the time at which the packet is generated within a PGA is not modeled. Investigating scheduled separation reduction requires this finer resolution, since minimum separation violations occur only when a packet is generated late enough that LOCC operations extend into the subsequent PGA. We therefore develop a separate, dedicated simulation package that resolves packet generation time within each PGA.

\section{Preliminaries}
\subsection{Network Model}

To optimize the scheduled separation chosen by the network scheduler
between consecutive deliveries, we
consider a GWR network supporting QI applications. Each end node
participates in entanglement generation and performs LOCC using the same hardware. End-to-end entangled links are
generated at rate~$\lambda$ with some minimum fidelity, abstracting the
network's interior delivery, the end nodes' consumption, and any hardware calibration latencies into a single
rate parameter. An application \texttt{App} is executed between a pair
of networked end nodes $\mathcal{N}=(N_1, N_2)$ and requires
$N^{\text{inst}}$ instances to complete, subject to a deadline
$t^{\text{expiry}}$. These parameters define an application
session~\cite{beauchamp2025modular},
\begin{equation}
    \mathcal{S} = (\texttt{session\_id}, \mathcal{N}, \texttt{App},
    N^{\text{inst}}, t^{\text{expiry}}).
\end{equation}
Each session submits a demand to the network specifying the identical
packet required for each instance. A packet is generated when the
required number of entangled links have been delivered and consumed
within a time window~$w$, where $w$ is constrained by the coherence time
of the end node's quantum memory. The network schedules PGAs for each
demand, and end nodes participate in all scheduled PGAs.



Following~\cite{arqon_arxiv}, we assume each instance's LOCC begin only after a packet is generated and that end nodes do not interrupt an executing instance to participate in entanglement generation. Successive packet generations not separated by at least the duration of LOCC violate the minimum separation constraint, and we assume the end nodes skip the subsequent PGA. In practice, LOCC also account for hardware calibration (e.g.\ ion cooling~\cite{liu2026long} or charge-resonance checks in color-center platforms~\cite{stolk2024metropolitan}) and software stack overheads such as processing~\cite{delle2025operating} and communication latencies~\cite{pompili2022experimental}.

\begin{figure}[t]
  \noindent\makebox[\columnwidth][l]{%
    \includegraphics[width=\columnwidth]{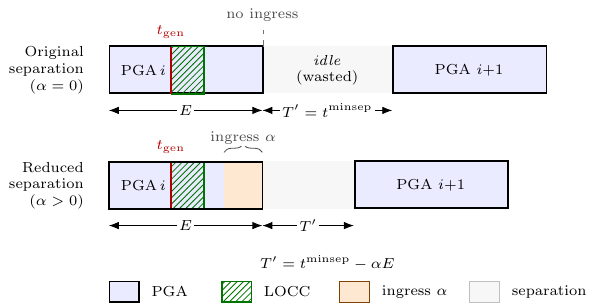}%
  }
  \caption{The ingress $\alpha$ is the fraction of PGA $i$ by which the
  scheduled separation is shortened, giving $T'=t^{\text{minsep}}-\alpha E$.
  At $\alpha=0$ (top) the full minimum separation is kept and the network
  idles after an early $t_{\text{gen}}$. At $\alpha>0$ (bottom) the
  separation shortens by $\alpha E$, reclaiming idle time. A packet
  generated within the shaded ingress region risks a minimum separation
  violation (Section~\ref{sec:relax_minsep}).\vspace{-8pt}
}
  \label{fig:minsep}
\end{figure}

\subsection{Performance Metrics} \label{sec:perf_metrics}
To quantify the effects of scheduling improvements on network resource
provisioning and end node session execution, we define three performance
metrics: the reduction in network service time~$R$, the absolute service
time savings $\Delta T_{\text{serv}}$, and the session execution
speedup~$S$. We first introduce the following quantities. For
session~$\mathcal{S}$, let $E$ denote the duration of a PGA,
$t^{\text{minsep}}$ the minimum time an end node requires for LOCC between successive consumptions, and $T' \leq t^{\text{minsep}}$
the scheduled separation between consecutive PGAs. Each instance requires
exactly one packet, and the network allocates a \emph{release time} of
$E + T'$ per instance, following the terminology of Arqon~\cite{arqon_arxiv}. Let $\mathcal{N}^{\text{scheduled}}$ denote the
total number of release times allocated by the network for a session.
Let $t_{\text{gen}}$ denote the time at which a packet is generated
within a successful PGA, and let $\tau$ denote the dimensionless release-time index at which the final
required packet is generated, counting both executed and skipped
PGAs.

The network service time $T_{\text{serv}} = \mathcal{N}^{\text{scheduled}}
\cdot (E + T')$ is the total duration over which network resources are
scheduled for a single session. The reduction in service time at a
reduced scheduled separation ($T'_1 < t^{\text{minsep}}$) relative to
the baseline at the full minimum separation ($T'_2 = t^{\text{minsep}}$)
is
\begin{equation} \label{eq:red_sched}
    R = 1 - \frac{T_{\text{serv},1}}{T_{\text{serv},2}},
\end{equation}
with corresponding absolute saving ${\Delta T_{\text{serv}} = R \cdot
T_{\text{serv},2}}$.

The service time upper bounds the session execution time,
\begin{equation}\label{eq:T-session}
    T_{\text{session}} = (\tau - 1) \cdot (E + T') + t_{\text{gen}},
\end{equation}
since not all allocated release times may be needed to generate the
required number of packets. When ${\tau = 1}$, the first PGA produces the
final required packet and ${T_{\text{session}} = t_{\text{gen}} \leq E+T'}$.
In the ideal case where every scheduled PGA succeeds, ${\tau =
N^{\text{inst}}}$ and the session completes at ${(N^{\text{inst}} - 1)
\cdot (E + T') + t_{\text{gen}}}$. The session execution speedup compares
$T_{\text{session}}$ at the reduced and baseline separations:
\begin{equation}\label{eq:speedup-ratio-agnostic}
    S \;=\;
    1 - \frac{T_{\text{session},1}}{T_{\text{session},2}}.
\end{equation}

These metrics depend only on the scheduling structure of repeated
entanglement deliveries and apply to any network
architecture in which persistent demands require successive
generation attempts separated by a minimum interval.

\subsection{Architecture}

To develop and evaluate the methods for the constrained scheduling problem, we adopt the demand and scheduling formalism of Arqon~\cite{arqon_arxiv}, which defines end node demands and how their corresponding packet generation probabilities and PGA execution times are calculated. Each session $\mathcal{S}$ submits a demand $d$ for entanglement generation to the network,
\begin{equation}
    d = (\mathfrak{p};\; t^{\text{minsep}};\; t^{\text{expiry}};\; N^{\text{inst}}; \epsilon_d^{\text{service}})
\end{equation}
with packet $\mathfrak{p} = (w, s, F)$ requesting $s$ entangled links
of minimum fidelity $F$ to be generated within a window $w$, and
$\epsilon_d^{\text{service}}$ defining the service probability
$(1 - \epsilon_d^{\text{service}})$ with which the demand must be
fulfilled.
Arqon constructs periodic schedules over fixed-length scheduling intervals, and each scheduled PGA attempts to generate a packet for end node consumption. To respect the minimum separation constraint, Arqon separates consecutive PGAs scheduled for the same demand by a minimum interval of $t^{\text{minsep}}$. The PGA duration $E(p)$ is the time required to generate a packet with probability at least $p$, determined by the entanglement generation rate along the end-to-end route: larger $p$ requires more attempts and therefore a longer $E(p)$. The PGA success probability $p^{\text{packet}}$ is then determined by the optimization 
\begin{equation}\label{eq:p-packet-opt}
    p^{\text{packet}} = \operatorname*{argmin}_{p}\; \mathcal{N}^{\text{SI}}(p, \epsilon_d^{\text{service}})\Big(E(p) + t^{\text{minsep}}\Big),
\end{equation}
where $\mathcal{N}^{\text{SI}}(p, \epsilon_d^{\text{service}})$ is the number of PGAs scheduled per scheduling interval for the demand to generate $N^{\text{inst}}$ packets with probability at least $(1 - \epsilon_d^{\text{service}})$. If all scheduled PGAs execute, the expected number of packets generated is
\begin{equation} \label{eq:min_serv}
    \mathbb{E}[P] = p^{\text{packet}} \, \mathcal{N}^{\text{executed}} \geq N^{\text{inst}}.
\end{equation}
\emph{Minimum service} is achieved when the number of packets generated is at least $N^{\text{inst}}$, and the service agreement between the network and an accepted demand guarantees that minimum service is provided with probability at least $(1 - \epsilon_d^{\text{service}})$.

Scheduling improvements must continue to satisfy the service agreement while also reducing the execution time of each session. The definitions of $T_{\text{serv}}$, $T_{\text{session}}$, and session speedup $S$ from Section~\ref{sec:perf_metrics} apply directly with $\mathcal{N}^{\text{scheduled}} = \mathcal{N}^{\text{SI}}$. Because Arqon allocates the same number of PGAs for a given demand regardless of the scheduled separation, $\mathcal{N}^{\text{SI}}$ is identical for both the original and reduced schedules. Expanding $T_{\text{serv}} = \mathcal{N}^{\text{SI}} \cdot (E + T')$ in Eq.~\eqref{eq:red_sched} shows the ratio $\mathcal{N}^{\text{SI}}_1 / \mathcal{N}^{\text{SI}}_2$ cancels, and the service time reduction and session speedup follow from Eqs.~\eqref{eq:red_sched} and ~\eqref{eq:speedup-ratio-agnostic}.

\subsection{Problem Statement}
Arqon satisfies the minimum separation constraint by scheduling consecutive PGAs with a separation of $t^{\text{minsep}}$. This assumes the worst-case generation time of $E$. In practice, packets are frequently generated before $E$, and the required separation depends on the generation time within the PGA. Arqon's conservative treatment of this constraint therefore overestimates the required scheduled separation, prolonging session execution time and reserving network resources that could otherwise be allocated to nodes ready to use them. Reducing the scheduled separation below $t^{\text{minsep}}$ shortens the release time and recovers this excess capacity. However, a packet generated sufficiently late in the PGA may violate the minimum separation constraint. This work addresses this tradeoff by computing the scheduled separation
that minimizes the expected session execution time, subject to the
expected packet count reaching $N^{\text{inst}}$ within the allocated
$\mathcal{N}^{\text{SI}}$ PGAs. 
The probabilistic minimum service guarantee
is not formally enforced by the optimization and is instead verified
empirically by simulation across all tested configurations.

\section{Methods} \label{sec:methods}
The following methods develop the analytical tools for computing the
optimal scheduled separation and are publicly available in~\cite{minsep_relaxation_2026}.

\subsection{Simulating Packet Generation}
We employ a Monte Carlo approach \cite{harrison2010introduction} to evaluate whether a packet is generated within a sliding window $w$ for a given demand $d$. 
The simulation, described by Algorithm~\ref{alg:sliding_window},
samples end-to-end entangled link generation as a Poisson process with
rate~$\lambda$ over the PGA duration~$E$. At each new link generation, the current sliding window is set to the earliest generation still within $w$ of the latest, tracked by the pointer $\text{ptr}_{\text{first}}$ into the arrivals list. If $s$ links accumulate within the window, the packet is generated at time $t_{\text{gen}}$ marking the final accumulated link. Otherwise, $E$ expires before $s$ links accumulate and the PGA fails to generate a packet.

We compare two parameter regimes of $\lambda$ and $w$ following recent demonstrations of entanglement distribution between trapped-ion nodes at metropolitan-scale distances. 
Liu~\textit{et al.}~\cite{liu2026long} represents the state-of-the-art (SOTA) in rate and coherence time, achieving ($\lambda = 2.22~\mathrm{Hz}, \ t_{\text{coh}} = 550~\mathrm{ms}$). We also compare this to an optimistic ion-trap regime of ($\lambda = 10~\mathrm{Hz}, \ t_{\text{coh}} = 4~\mathrm{s}$) \cite{arqon_arxiv}. The window $w$ is set to $t_{\text{coh}} / 2$, reserving the remaining coherence time for local quantum operations. The minimum separation is set to $t^{\text{minsep}} = t_{\text{coh}} / 2 + t_{\text{LOCC}}$, where $t_{\text{LOCC}} = 0.25~\mathrm{s}$ accounts for local operations and classical communication between nodes~\cite{delle2025operating, van2025qoala}, and trapped ion cooling latencies~\cite{ hucul2015modular}. Qubit decoherence after delivery is not modeled, as we assume the window $w$ is chosen to ensure that delivered qubits remain coherent throughout instance execution.

\begin{algorithm}[] \label{alg:sliding_window}
\caption{Packet generation within sliding window}
\KwIn{$\lambda$, $t_{\text{start}}$, $E$, $w$, $s$}
\KwOut{Success flag and $t_{\text{gen}}$, the time $s$ links accumulate within $w$}
Initialize $t \leftarrow t_{\text{start}}$, \ $\text{ptr}_{\text{first}} \leftarrow \texttt{nil}$, \ $\text{arrivals} \leftarrow [\,]$\;
\While{$t < E$}{
    $t \leftarrow t + \text{Exp}(1/\lambda)$\;
    \If{$t \geq E$}{
        \textbf{break}\;
    }
    Append $t$ to $\text{arrivals}$\;
    \For{each $t_i \in \text{\emph{arrivals}}$}{
        \If{$t - t_i < w$}{
            $\text{ptr}_{\text{first}} \leftarrow \text{index of } t_i$\;
            \textbf{break}\;
        }
    }
    \If{$|\text{\emph{arrivals}}| - \text{\emph{ptr}}_{\text{\emph{first}}} \geq s$}{
        $t_{\text{gen}} \leftarrow t$\;
        \Return{\texttt{\emph{true}}, $t_{\text{\emph{gen}}}$}\;
    }
}
\Return{\texttt{\emph{false}}, \texttt{\emph{nil}}}
\end{algorithm}

\subsection{Relaxing a Demand's Scheduled Separation} \label{sec:relax_minsep}
We consider reducing the scheduled separation between consecutive PGAs from the demand's minimum separation $t^{\text{minsep}}$ to $T' \leq t^{\text{minsep}}$, described by Figure \ref{fig:minsep}. A \emph{minimum separation violation} occurs when the interval between packet generation in PGA~$i$ and the start of PGA~$i{+}1$ is less than the time required for LOCC $t^{\text{minsep}}$, formally ${T' + (E - t_{\text{gen}}) < t^{\text{minsep}}}$. We restrict our analysis to the \emph{single-skip regime}, in which at most one consecutive PGA is ever skipped. In the worst case, a packet generated at time $E$ of PGA~$i$ triggers a skip of PGA~$i{+}1$, and the time from that packet to the start of PGA~$i{+}2$ is $2T' + E$. Reducing $T'$ further would leave insufficient time for PGA~$i{+}2$ and trigger cascading skips. The single-skip guarantee therefore requires
\begin{equation}\label{eq:single-skip}
    2T' + E \geq t^{\text{minsep}},
\end{equation}
yielding a lower bound on the reduced separation of ${(t^{\text{minsep}} - E)/2 \leq T'}$. Since $T'$ is a scheduled time interval and therefore non-negative, the admissible range is
\begin{equation}\label{eq:admissible-range}
    \max\!\left(0,\; \frac{t^{\text{minsep}} - E}{2}\right) \leq T' \leq t^{\text{minsep}}.
\end{equation}
At the upper bound, the original schedule is recovered. At the lower bound, the single-skip constraint~\eqref{eq:single-skip} ensures the worst-case violation skips only PGA~$i{+}1$. Further reduction would permit cascading skips.

We parameterize the reduction of the scheduled separation by the PGA
ingress $\alpha = (t^{\text{minsep}} - T')/E \in [0,1]$. A packet
generated within the \emph{violation zone} $((1-\alpha)E, E]$ causes a
minimum separation violation. The conditional probability of this event
given successful packet generation is
\begin{equation}\label{eq:pmv}
    p^{\text{mv}}(\alpha) = \mathbb{P}[t_{\text{gen}} \in ((1-\alpha)E,
    E] \;\mid\; t_{\text{gen}} \leq E].
\end{equation}
Under our earlier assumption that a minimum separation violation causes
the end node to skip the subsequent PGA, $p^{\text{mv}}(\alpha)$ is also
the probability that an end node skips the next PGA. The distribution of
packet generation times within a PGA admits no closed-form expression,
so $p^{\text{mv}}(\alpha)$ cannot be evaluated analytically. The Markov
chain session model in the following subsection requires $p^{\text{mv}}$
at many values of $\alpha$ during the constrained optimization, which
would be prohibitively expensive to evaluate by fresh Monte Carlo
simulation, particularly at runtime by a network scheduler. To avoid
this cost, we precompute $p^{\text{mv}}$ into a lookup table by sweeping
$p^{\text{packet}}$ and~$\alpha$ with the Monte Carlo simulator of
Algorithm~\ref{alg:sliding_window}. Each table is computed once per
physical configuration $(\lambda, w, s)$ and reused across all demands
sharing those parameters. For each configuration, we collect $N=10^4$
successful trials via rejection sampling, and $p^{\text{mv}}$ is the
fraction in which the packet is generated within the violation zone.
Once the table is available, the Markov chain model and the optimization
over~$\alpha$ reduce to table lookups and closed-form expressions,
enabling lightweight runtime evaluation within the network scheduler.

\subsection{Modeling Session Performance with Reduced Scheduled Separation} \label{sec:sess_perf_methods}

To evaluate session performance under reduced scheduled separation, we
simulate sessions by executing $\mathcal{N}^{\text{SI}}$ scheduled
PGAs subject to the skip behavior introduced in
Section~\ref{sec:relax_minsep}, in which a minimum separation violation
causes the end node to skip the subsequent PGA. Although demands are generally long-lived and
served across multiple scheduling intervals until $t^{\text{expiry}}$, our
analysis focuses on a single session, which isolates the tradeoff between
reduced separation and skip penalties. For each executed PGA, sampled link generations determine whether a packet is generated and whether $t_{\text{gen}}$ falls in the violation zone $((1-\alpha)E,\, E]$. The session terminates when $N^{\text{inst}}$ packets have been generated or when all $\mathcal{N}^{\text{SI}}$ PGAs are exhausted. For completed sessions, $T_{\text{session}}$ takes the form of Eq.~\eqref{eq:T-session}, with $\tau$ the release time in which the final packet is generated. Sessions that do not generate $N^{\text{inst}}$ packets within the allocation are assigned ${T_{\text{session}} = \mathcal{N}^{\text{SI}} \cdot (E + T'(\alpha))}$. We estimate the expected session times as sample means over all trials, both completed and incomplete. The average simulated speedup at a given ingress is
\begin{equation}\label{eq:speedup-sim}
    \overline{S}(\alpha)
    \;=\;
    1 - \frac{\mathbb{E}\bigl[(\tau(\alpha){-}1)(E + T'(\alpha))
              + t_{\text{gen}}(\alpha)\bigr]}
             {\mathbb{E}\bigl[(\tau(0){-}1)(E + t^{\text{minsep}})
              + t_{\text{gen}}(0)\bigr]}.
\end{equation}

We define $\overline{S}(\alpha)$ as one minus the ratio of expected
session execution times, not the expectation of the per-session speedup of
Eq.~\eqref{eq:speedup-ratio-agnostic}. We
report the former to match the mean-based quantities predicted by
the Markov chain model developed below. To predict session performance without repeated simulation, we introduce a Markov chain for the expected number of PGAs executed out of the total scheduled. We introduce several quantities to define the state space. Let $N_k$ be the cumulative number
of executed PGAs after $k$ scheduled steps. Each scheduled step is labeled
as either executed ($\overline{G}$) or skipped due to a minimum separation
violation at the preceding step ($G$), and $S_k \in \{\overline{G}, G\}$
denotes this label at step $k$.

The session is then modeled as the Markov chain ${M_k = (N_k, S_k)}$ on the
state space $\mathbb{N} \times \{\overline{G}, G\} \cup \{(0, \overline{G})\}$. The
transition structure reflects the skip rule directly. In state
$\overline{G}$, the scheduled PGA is executed, and it triggers a skip of
the next step only if both hold: a packet is generated,
which occurs with probability $p^{\text{packet}}$, and its generation time
falls in the violation zone, which occurs with
conditional probability $p^{\text{mv}}$. The transition
$\overline{G} \to G$ therefore occurs with unconditional probability
$p_r = p^{\text{packet}} \cdot p^{\text{mv}}$. Oherwise the state remains
in $\overline{G}$. From $G$, the state returns deterministically to $\overline{G}$, since a skipped step triggers no minimum separation violation, and the single-skip regime of Eq.~\eqref{eq:single-skip} prevents any repeated ones.

The sub-chain $S_k$ forms a reduced Markov chain on the state space
$\{\overline{G}, G\}$, which describes the transitions between executed
and skipped steps and can be used to compute the $k$-step transition
probabilities. Its transition matrix is
\begin{equation} \label{eq:transition}
    P_S = \begin{bmatrix} 1 - p_r & p_r \\ 1 & 0 \end{bmatrix},
\end{equation}
where rows and columns are ordered $(\overline{G},\, G)$, from which the $k$-step transition probabilities follow by induction. In particular, the probability of being in the skip state~$G$ at step~$k$, starting from~$\overline{G}$, is 
\begin{equation} \label{eq:fk}
    f_k(p_r) = p_r \sum_{i=0}^{k-1} (-p_r)^i.
\end{equation}
 Conditioning $N_k$ on $S_{k-1}$ yields the recurrence
\begin{equation} 
    \mathbb{E}[N_k] = \mathbb{E}[N_{k-1}] + 1 - f_{k-1}(p_r),
\end{equation}
where $f_k$ is as in \eqref{eq:fk}. This resolves to the closed form
\begin{equation}
    \mathbb{E}[N_k] = \sum_{i=0}^{k-1} (k - i)(-p_r)^i.
\end{equation}
The expected number of packets generated after $k$ scheduled steps is then $\mathbb{E}[P_k] = p^{\text{packet}} \cdot \mathbb{E}[N_k]$, since each executed PGA independently produces a packet with probability~$p^{\text{packet}}$. We define $k^*$ as the smallest~$k$ such that $\mathbb{E}[P_k] \geq N^{\text{inst}}$. An outline of the full proof is provided in Appendix \ref{app:markov}. 

In any individual session, however, variance in both the number of executed PGAs and the per-PGA packet outcomes causes~$\tau$, the release time the session completes, to differ from~$k^*$. Because $P_k$ has positive variance around its mean, a
fraction of simulated sessions complete before step~$k^*$, such that
$\mathbb{E}[\tau] \leq k^*$. The quantity
${k^*(\alpha) \cdot (E + T'(\alpha))}$ therefore provides a conservative
upper bound on~${\mathbb{E}[T_{\text{session}}(\alpha)]}$. The minimum session speedup is then
\begin{equation}\label{eq:speedup-ratio}
    \tilde{S}(\alpha^*, k^*) \;=\;
    1 - \frac{k^*(\alpha^*)}{k^*(0)}
    \cdot
    \frac{E + T'(\alpha^*)}{E + t^{\text{minsep}}},
\end{equation}
which lower bounds the average simulation speedup ${\overline{S}(\alpha^*) \geq \tilde{S}(\alpha^*, k^*)}$. 

Three constraints bound the admissible PGA ingress~$\alpha$ and define
the feasible region of the optimization. First, the scheduled
separation cannot be reduced below zero, since there is no remaining
interval to shorten. This \emph{geometric constraint} requires
${T'(\alpha) \geq 0}$, equivalently ${\alpha \leq t^{\text{minsep}}/E}$.
Second, the single-skip constraint~\eqref{eq:single-skip} requires
${T'(\alpha) \geq \max(0,\, (t^{\text{minsep}} - E)/2)}$, ensuring that
the worst-case violation skips at most one consecutive PGA. When
${E < t^{\text{minsep}}}$, this prevents~$T'$ from being reduced to
zero. Finally, the service feasibility constraint requires
$k^*(\alpha) \leq \mathcal{N}^{\text{SI}}$, ensuring the expected packet count reaches~$N^{\text{inst}}$ within the allocated PGAs. We
define~$\alpha^*$ as the ingress that minimizes the Markov session
time ${T_{\text{session}}(\alpha) = k^*(\alpha) \cdot (E + T'(\alpha))}$
over this feasible region.

We remark that $k^*$ enforces an expectation-based analog of the probabilistic service guarantee defined in Eq.~\eqref{eq:min_serv}. While minimum
service requires that $N^{\text{inst}}$ packets are generated with
probability at least $(1 - \epsilon_d^{\text{service}})$, the Markov chain
tracks only the expected packet count and does not bound the probability
that an individual session falls short. The service feasibility constraint
$k^*(\alpha) \leq \mathcal{N}^{\text{SI}}$ is therefore necessary but not
sufficient for the probabilistic guarantee, and we verify empirically with simulation that
minimum service is maintained at the optimal ingress across all tested
configurations.

While the speedup lower bound in Eq.~\eqref{eq:speedup-ratio} can be evaluated numerically from~$k^*$, it does not reveal the structure of the tradeoff between reduced release time and skip penalties.
To decompose the speedup into scheduling quantities, we derive a closed-form approximation of the ratio $k^*(\alpha^*)/k^*(0)$ using the stationary distribution of the reduced Markov chain.
The stationary distribution $\boldsymbol{\pi} = (\pi_{\overline{G}},\, \pi_{G})$ satisfies ${\boldsymbol{\pi} P_S = \boldsymbol{\pi}}$ and ${\pi_{\overline{G}} + \pi_{G} = 1}$, which yields
\begin{equation}\label{eq:stationary}
    \pi_{\overline{G}} = \frac{1}{1 + p_r}, \qquad
    \pi_{G} = \frac{p_r}{1 + p_r}.
\end{equation}
In steady state, a fraction $\pi_{\overline{G}}$ of scheduled PGAs result in an executed PGA. Since each executed PGA produces a packet with probability~$p^{\text{packet}}$ regardless of the skip dynamics, the number of executions required to generate~$N^{\text{inst}}$ packets is the same at any ingress. In the stationary limit, the number of scheduled steps at ingress~$\alpha^*$ therefore satisfies
\begin{equation}\label{eq:k-star-inflation}
    \frac{k^*(\alpha^*)}{1 + p_r(\alpha^*)} \;=\; k^*(0),
\end{equation}
and $k^*(\alpha^*) = k^*(0) \cdot (1 + p_r(\alpha^*))$. The relation is exact in the stationary limit and applies to finite sessions once they are long enough for the chain to converge, as is the case for all tested configurations with $N^{\text{inst}} = 10^3$ averaged over $10^3$ iterations. For shorter sessions, Eq.~\eqref{eq:k-star-inflation} should be replaced by numerical evaluation.

\begin{figure*}[tp]
    \centering
    \begin{minipage}[t]{0.48\textwidth}
        \centering
        \includegraphics[width=\textwidth]{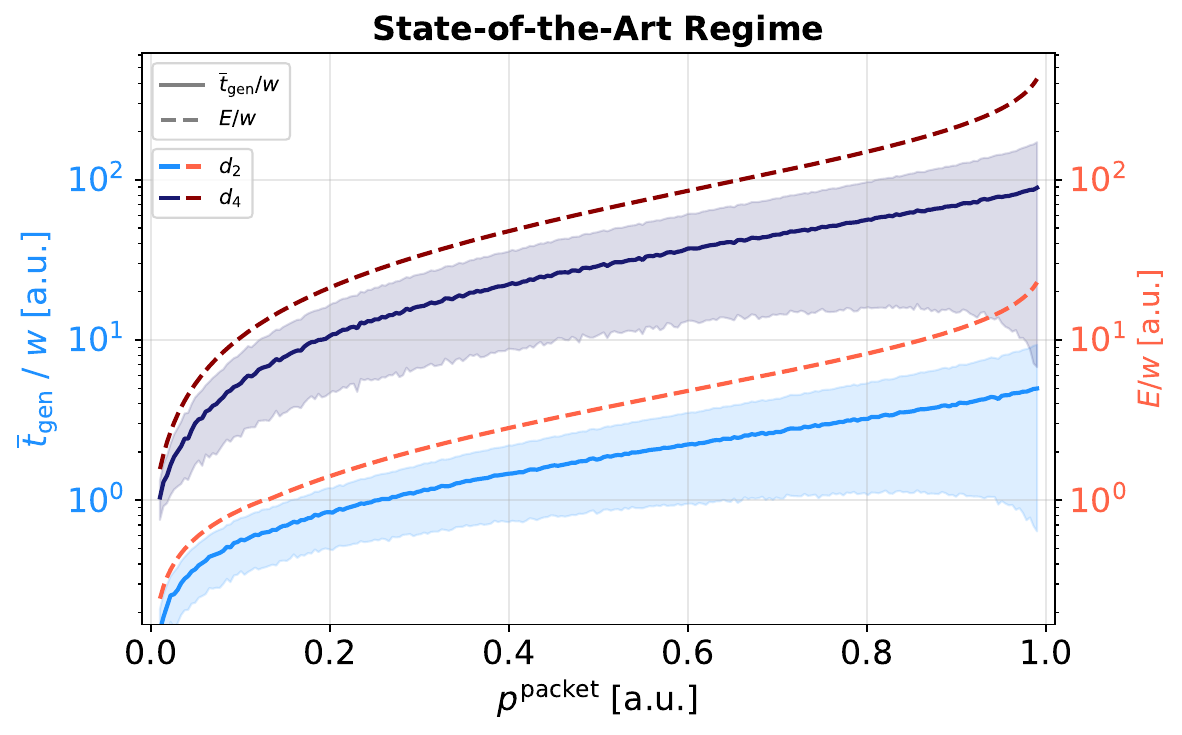}
    \end{minipage}
    \hfill
    \begin{minipage}[t]{0.48\textwidth}
        \centering
        \includegraphics[width=\textwidth]{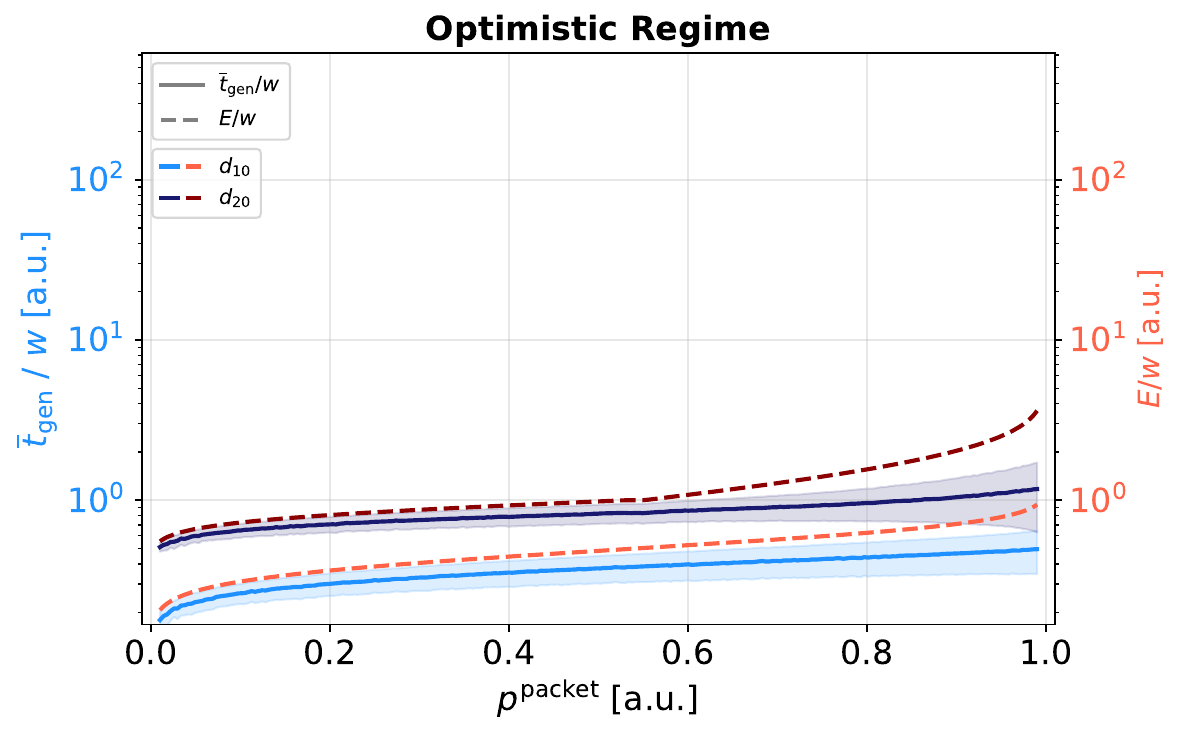}
    \end{minipage}
    \caption{Average packet generation time (solid, left axis) and PGA duration (dashed, right axis) as a function of packet generation probability
$p^{\text{packet}}$ for the state-of-the-art regime
(left; $\lambda = 2.22~\mathrm{Hz}$, $w = 275~\mathrm{ms}$) and
optimistic regime (right; $\lambda = 10~\mathrm{Hz}$, $w = 2~\mathrm{s}$).
Both quantities are normalized by the window length~$w$ to enable comparison
across demands with different PGA durations. Demands are labeled by their
requested link count~$s$, with $z$-scores
$z = (s - \lambda w)/\sqrt{\lambda w}$ of $1.8$ and $4.3$ for~$d_2$
and~$d_4$ (SOTA), and $-2.2$ and $0$ for~$d_{10}$ and~$d_{20}$
(optimistic). Shaded regions indicate $\pm$ one standard deviation
over $10^4$ trials per point. All demands use $N^{\text{inst}} = 1$ and
$\epsilon_d^{\text{service}} = 10^{-5}$.} \label{fig:p_packet_vs_t_gen}
\end{figure*}

Substituting Eq.~\eqref{eq:k-star-inflation} into Eq.~\eqref{eq:speedup-ratio} yields a closed-form speedup expressed only by the release time reduction $R(\alpha^*)$ and the unconditional skip probability $p_r(\alpha^*)$. Because $\mathcal{N}^{\text{SI}}$ is held fixed across scheduled separations, the service time reduction in Eq.~\eqref{eq:red_sched} simplifies to the release time reduction ${R(\alpha) = 1 - (E + T'(\alpha))/(E + t^{\text{minsep}})}$, and the stationary speedup becomes
\begin{align}\label{eq:speedup-decomp}
    \tilde{S}(\alpha^*)
    &\;=\;
    1 - \bigl(1 + p_r(\alpha^*)\bigr)
    \cdot \bigl(1 - R(\alpha^*)\bigr) \nonumber\\[4pt]
    &\;=\;
    \underbrace{R(\alpha^*)}_{\text{release time reduction}}
    \;-\;
    \underbrace{p_r(\alpha^*) \cdot
    \bigl(1 - R(\alpha^*)\bigr)}_{\text{skip penalty}}.
\end{align}
The first term, $R(\alpha^*)$, captures the fraction of session time savings if no PGAs were skipped. The second term captures the fraction of that savings lost to the additional PGAs executed to compensate for skipped PGAs. The speedup is positive only when the release time reduction exceeds the skip penalty. Evaluating $\tilde{S}(\alpha^*)$ requires no simulation, only the $p^{\text{mv}}(\alpha^*)$ lookup table and ${T'(\alpha^*) = t^{\text{minsep}} - \alpha^* E}$ computed from demand parameters. Determining the optimal $\alpha^*$ requires evaluating the service feasibility constraint $k^*(\alpha) \leq \mathcal{N}^{\text{SI}}$ through the Markov chain recurrence, since the decomposition alone cannot predict whether the session completes in time.

\section{Evaluation} \label{sec:evaluation}

We evaluate the proposed methods in three stages. We first characterize the distribution of packet generation times within a PGA and the resulting lookup tables of minimum separation violation probabilities that drive the optimization, then quantify the single-session speedup, and finally extend to round-robin scheduling of two co-scheduled sessions.

\subsection{Simulating Packet Generation Time Distributions} \label{sec:p_time}

To evaluate scheduled separation relaxation, the distribution of packet generation times within a PGA must first be understood, since packets arriving near the end of the PGA are more likely to cause minimum separation violations. Demands~$d_s$ are indexed by the requested packet's link count ~$s$.
Fig.~\ref{fig:p_packet_vs_t_gen} shows~$d_2$ and~$d_4$ in the SOTA
regime and~$d_{10}$ and~$d_{20}$ in the optimistic regime.

Fig.~\ref{fig:p_packet_vs_t_gen} plots $\overline{t}_{\text{gen}}/w$ against $p^{\text{packet}}$, with $E/w$ on the secondary axis; normalization by~$w$ enables comparison across demands. The $z$-score $z = (s - \lambda w) / \sqrt{\lambda w}$ indicates within how many standard deviations the requested link count lies from the expected number of arrivals per window. When $z \gg 0$, a single window is unlikely to generate enough links, so the PGA spans many overlapping windows ($E \gg w$). Generation times then spread across a substantial portion of~$E$, as seen for $d_2$ at $z = 1.8$ and $d_4$ at $z = 4.3$. When $z \leq 0$, a single window suffices ($E \leq w$) and packets are likely generated within roughly one window scan ($d_{10}$ at $z = -2.2$).

\subsection{Demand Lookup Tables}\label{subsec:lookup-tables}
\begin{figure*}[t]
    \centering
    \begin{minipage}[t]{0.48\textwidth}
        \centering
        \includegraphics[width=\textwidth]{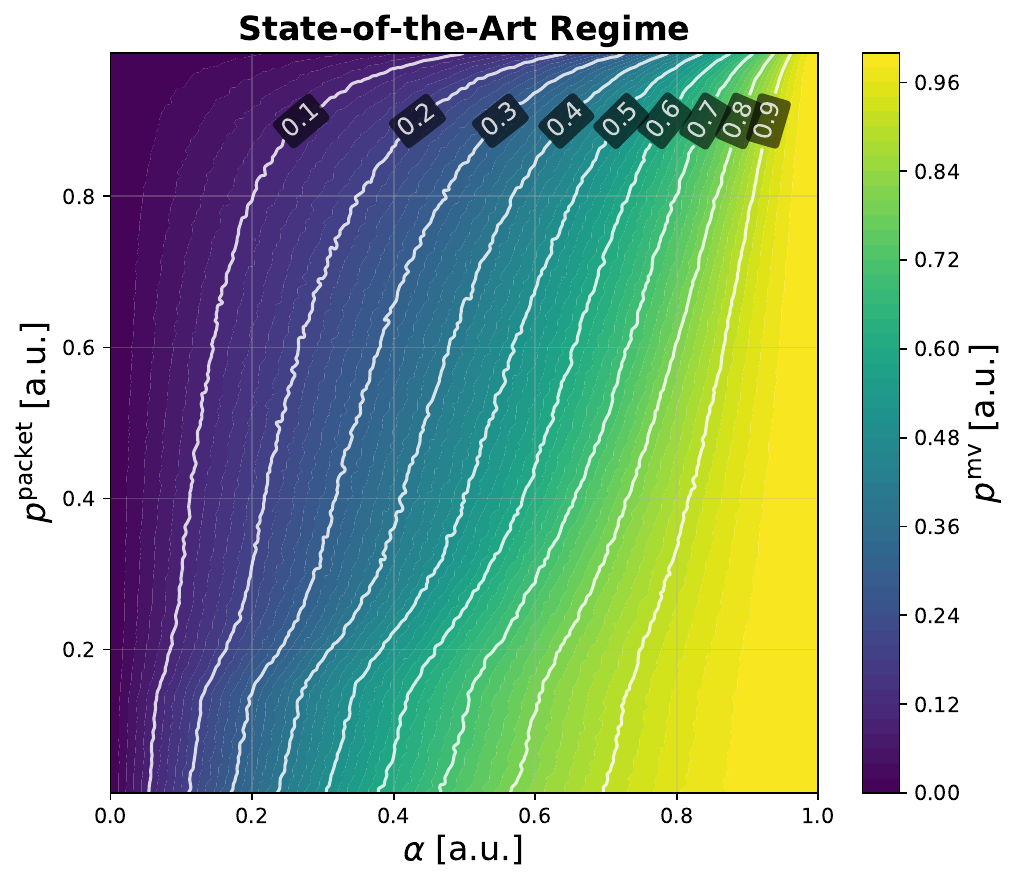}
    \end{minipage}
    \hfill
    \begin{minipage}[t]{0.48\textwidth}
        \centering
        \includegraphics[width=\textwidth]{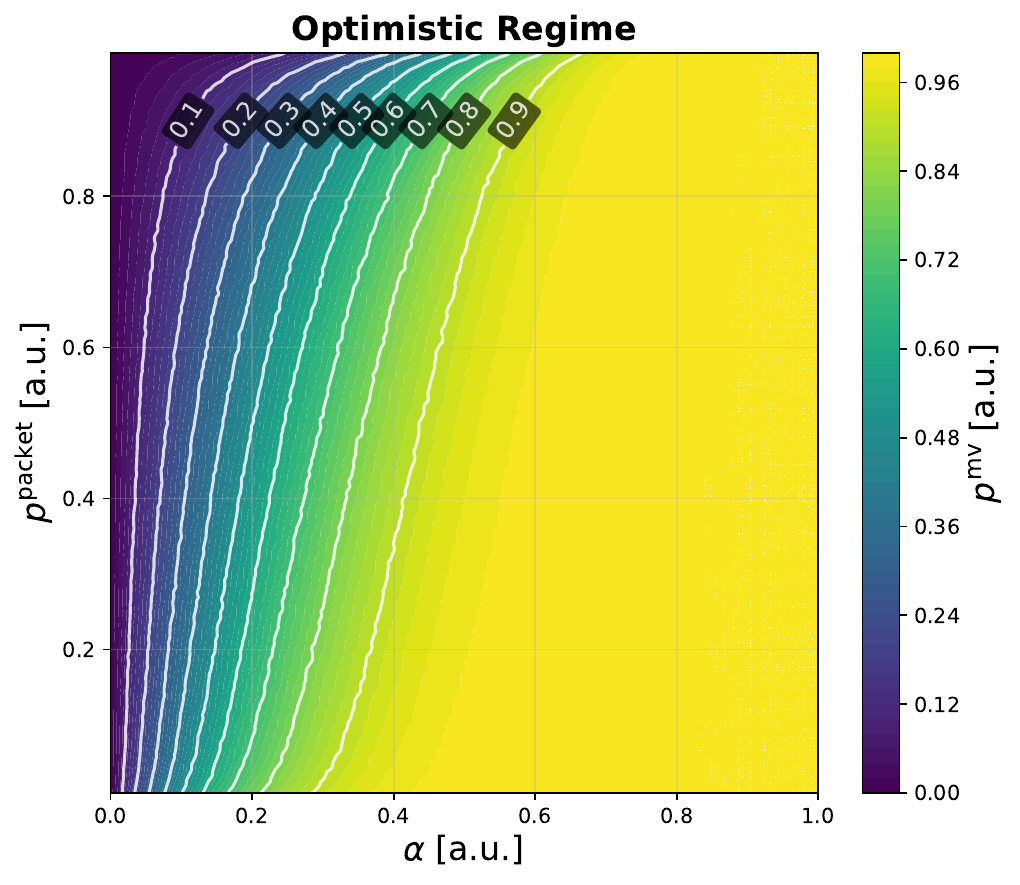}
    \end{minipage}
    \caption{Probability $p^{\text{mv}}$ of a minimum separation violation,
conditioned on successful packet generation. The surfaces are shown as
a function of PGA ingress~$\alpha$ and packet generation
probability~$p^{\text{packet}}$. Two configurations are compared:
requested links $s = 2$ in the state-of-the-art regime (left;
$\lambda = 2.22~\mathrm{Hz}$, $w = 275~\mathrm{ms}$, $z = 1.8$) and
$s = 10$ in the optimistic regime (right; $\lambda = 10~\mathrm{Hz}$,
$w = 2~\mathrm{s}$, $z = -2.2$). Each grid point is estimated from $10^4$ successful trials obtained by rejection sampling. White contours are spaced at intervals of $0.1$. These surfaces are computed once per physical configuration $(\lambda, w, s)$ and stored as lookup tables for use by the Markov chain session model. }
    \label{fig:lookup_tables} \vspace{-2mm}
    
\end{figure*}

Since no closed-form expression exists for the packet generation time within a PGA, the conditional violation probability $p^{\text{mv}}$ must be obtained by sampling. Fig.~\ref{fig:lookup_tables} presents $p^{\text{mv}}$ as a function of~$\alpha$ and~$p^{\text{packet}}$ for one representative configuration from each regime:~$d_2$ ($z = 1.8$) in the SOTA regime and~$d_{10}$ ($z = -2.2$) in the optimistic regime. The two surfaces exhibit qualitatively distinct contour structures. In the SOTA regime ($E \gg w$),~$p^{\text{mv}}$ contours disperse along the~$\alpha$ axis, allowing the ingress to encompass a substantial fraction of the PGA before minimum separation violations dominate. In the optimistic regime ($E \leq w$), the contours concentrate at low~$\alpha$ and~$p^{\text{mv}}$ rises steeply with even small ingress fractions. Increasing~$p^{\text{packet}}$ at fixed~$\alpha$ reduces~$p^{\text{mv}}$ in both regimes, since a longer PGA shifts generation times earlier.

These surfaces confirm that the $z$-score governs $p^{\text{mv}}$'s sensitivity to ingress: low-$z$ demands reach high minimum separation violation probabilities at smaller ingress and therefore require more conservative separation reductions. Each surface is computed once for a given physical configuration $(\lambda, w, s)$ and stored as a lookup table that the Markov chain session model queries during optimization. Because the table is precomputed, evaluating~$p^{\text{mv}}$ at any candidate $(\alpha, p^{\text{packet}})$ requires only interpolation rather than repeated simulation. In regimes where $\lambda w$ is large, a unit change in~$s$ produces only a small shift in $z$-score ($\Delta z = 1/\sqrt{\lambda w}$), suggesting that tables computed for nearby~$s$ values may provide reasonable approximations without recomputation. Quantifying the interpolation error across configurations is left to future work.

\subsection{Single Session Speedup}\label{subsec:session-speedup}
To evaluate session performance under reduced scheduled separation, demands~$d_s$ from each regime are configured with $N^{\text{inst}} = 10^3$ and $\epsilon_d^{\text{service}} = 0.5$, and the associated sessions are simulated over $10^3$ trials. We restrict the analysis to feasible demands with $z$-scores below~$5$, namely $d_2$,~$d_3$,~$d_4$ in the SOTA regime and~$d_{10}$,~$d_{20}$,~$d_{28}$,~$d_{34}$ in the optimistic regime. Of these,~$d_{10}$ yields no speedup at any feasible ingress, since its high $p^{\text{packet}} = 0.951$ causes the skip penalty to dominate the admissible range. The remaining six demands achieve positive speedup; for each, Table~\ref{tab:session_speedup} reports the optimal ingress~$\alpha^*$, the Markov minimum speedup bound~$\tilde{S}(\alpha^*, k^*)$, and the sample mean simulated speedup~$\overline{S}(\alpha^*)$ along with its 95\% confidence interval over $10^3$ trials. Across these configurations, $T_{\text{session}}(\alpha)$ is strictly decreasing in~$\alpha$ over the feasible region, so~$\alpha^*$ coincides with the largest feasible ingress. We remark that the stationary approximation~$\tilde{S}(\alpha^*)$ agrees with $\tilde{S}(\alpha^*, k^*)$ to within $0.05$ percentage points across all configurations, confirming that the chain converges well within the session length at $N^{\text{inst}} = 10^3$. The Markov bound $\tilde{S}(\alpha^*, k^*)$ underestimates the simulated speedup $\overline{S}(\alpha^*)$ by $0.1$ to $0.7$ percentage points, remaining tight enough to be predictive without repeated simulation.

Whether any speedup exists depends on $p^{\text{packet}}$. Since $p_r = p^{\text{packet}} \cdot p^{\text{mv}}(\alpha)$, a high $p^{\text{packet}}$ makes the skip penalty large at every~$\alpha$, and when $p^{\text{packet}}$ is near unity the penalty wipes out any gain across the admissible range, as observed for~$d_{10}$. When speedup is possible, its size is controlled by~$z$ through $R(\alpha^*)$. The geometric ceiling on~$R(\alpha^*)$ is ${t^{\text{minsep}} / (E + t^{\text{minsep}})}$, which shrinks with~$z$ as~$E$ grows relative to~$t^{\text{minsep}}$. However, the ceiling also becomes more reachable as~$z$ grows: longer PGAs spread generation times out, so~$p^{\text{mv}}$ rises more slowly with~$\alpha$ and the optimizer can push~$\alpha^*$ further before the skip penalty dominates. Near $z = 0$ ($d_{20}$), the ceiling is large at $42.2\%$,
but~$p^{\text{mv}}$ rises quickly and holds $R(\alpha^*)$ to $2.8\%$,
yielding a speedup of at least $0.8\%$. At the opposite extreme ($d_4$),
the ceiling shrinks to $7.6\%$ but is reached by $R(\alpha^*)$,
yielding a speedup of at least $6.1\%$.

Across all configurations, simulations confirm that at least~$N^{\text{inst}}$ packets are generated with probability $1 - \epsilon_d^{\text{service}}$ at~$\alpha^*$, though the analytical framework does not guarantee service delivery. Formalizing the relationship between~$\alpha^*$ and $\epsilon_d^{\text{service}}$ is left to future work.

\subsection{Multi-Session Speedup}
When multiple sessions share the network by round-robin scheduling, the PGAs of other sessions provide natural separation between consecutive same-session PGAs. For session $i$, this separation is ${g_i = T'_i + \sum_{j \neq i}(E_j + T'_j)}$. When ${g_i \geq t^{\text{minsep}}_i}$, no minimum separation violations occur and $T'_i$ can be fully eliminated. Otherwise, $\alpha^*_i$ still sets the scheduled separation through $T'_i = t^{\text{minsep}}_i - \alpha^*_i E_i$. However, only the residual $\alpha^*_{\text{eff},i} = (t^{\text{minsep}}_i - g_i)/E_i \leq \alpha^*_i$ contributes to violations through the violation zone $((1-\alpha^*_{\text{eff},i})E_i, E_i]$. We restrict our analysis to $n = 2$ co-scheduled sessions, which eliminates skip penalties for nearly all demand configurations demonstrating single-session speedup, since the no-violation condition reduces to $E \geq t^{\text{minsep}}$. The analysis extends to larger $n$ by replacing the single intervening PGA with the sum of the other $n-1$ intervening PGA durations.

We consider co-scheduling demand $d_{20}$ with demand $d_2$ under round-robin. The round-robin cycle has duration ${E_{20} + T'_{20} + E_2 + T'_2}$, with each $T'_i = t^{\text{minsep}}_i - \alpha^*_i E_i$. The cycle reduction $R_{\text{cyc}}$ is the fractional reduction of the optimized cycle relative to the baseline cycle in which every session uses its full minimum separation. We fix $\alpha^*_2 = 0.431$ at its geometric maximum, which eliminates $T'_2$; this is jointly optimal because $\alpha^*_{\text{eff},2} = 0$ leaves no skip penalty. Optimizing $d_{20}$ then yields $\alpha^*_{20} = 0.442$, giving an inter-PGA separation of $2.1$~s ($t^{\text{minsep}}_{20} = 2.25$~s) and a residual $\alpha^*_{\text{eff},20} = 0.047$.

\begin{table}[t!]
\centering
\caption{Single-Session performance at optimal ingress $\alpha^*$}
\label{tab:session_speedup}
\renewcommand{\arraystretch}{1.3}
\resizebox{\columnwidth}{!}{%
\begin{tabular}{l c c c c c c}
\toprule
Session & $z$ & $\alpha^*$ & $\tilde{S}(\alpha^*, k^*)$ [\%] & $\overline{S}(\alpha^*)$ [\%] & $R(\alpha^*)$ [\%] & $\Delta T_{\text{serv}}$ \\
\midrule
~$d_{2}$ & 1.8 & 0.063 & 2.0 & $2.7 \pm 0.1$ & 4.4 & 2.3~min \\
~$d_{3}$ & 3.1 & 0.095 & 5.2 & $5.8 \pm 0.2$ & 7.9 & 11.1~min \\
~$d_{4}$ & 4.3 & 0.082 & 6.1 & $6.2 \pm 0.2$ & 7.6 & 41.9~min \\
\midrule
~$d_{20}$ & 0.0 & 0.049 & 0.8 & $1.1 \pm 0.1$ & 2.8 & 3.2~min \\
~$d_{28}$ & 1.8 & 0.062 & 2.5 & $3.1 \pm 0.1$ & 4.9 & 18.1~min \\
~$d_{34}$ & 3.1 & 0.069 & 5.1 & $5.3 \pm 0.2$ & 6.5 & 2.8~hr \\
\bottomrule
\end{tabular}%
}
\end{table}

Two distinct mechanisms drive the multi-session speedup. First, the round-robin cycle replaces the
single-session release time $E + T'$, and the release time
reduction~$R(\alpha^*)$ is correspondingly replaced by the cycle reduction
$R_{\text{cyc}}$ defined above. Second, only the residual ingress contributes
to minimum separation violations, so the skip probability~$p_{r,i}$ is
evaluated at~$\alpha^*_{\text{eff},i}$ rather than at~$\alpha^*_i$.
Combining both effects, the per-session speedup admits the two
analytical forms developed for the single-session case: the Markov
lower bound on speedup $\tilde{S}_i(\alpha_i^*, k^*)$ analogous to
Eq.~(\ref{eq:speedup-ratio}), and its closed-form stationary
approximation $\tilde{S}_i(\alpha^*_{\text{eff},i})$ analogous to Eq.~(\ref{eq:speedup-decomp}). The Markov
bound is what Table~\ref{tab:multi_session_speedup} reports, and is
computed by evaluating the chain of
Section~\ref{sec:sess_perf_methods} with transition probability
$p_{r,i}(\alpha^*_{\text{eff},i})$ to obtain $k^*_i$, then substituting
into the multi-session analog of Eq.~(\ref{eq:speedup-ratio}). The
stationary approximation, obtained from the Markov bound by
substituting the stationary approximation of
Eq.~(\ref{eq:k-star-inflation}), gives a closed-form expression for
$\tilde{S}_i$ measured against the baseline cycle:
\begin{align}\label{eq:multi-speedup-decomp}
    \tilde{S}_i(\alpha^*_{\text{eff},i})
    &\;\approx\;
    1 - \bigl(1 + p_{r,i}(\alpha^*_{\text{eff},i})\bigr)
    \cdot \bigl(1 - R_{\text{cyc}}\bigr) \nonumber\\[4pt]
    &\;=\;
    \underbrace{R_{\text{cyc}}}_{\text{cycle reduction}}
    \;-\;
    \underbrace{p_{r,i}(\alpha^*_{\text{eff},i}) \cdot
    \bigl(1 - R_{\text{cyc}}\bigr)}_{\text{skip penalty}},
\end{align}
with analogous decomposition of scheduling quantities as
Eq.~(\ref{eq:speedup-decomp}). We remark that the approximation agrees with the
Markov bound to within $0.01$ percentage points.

The results of co-scheduling are reported in Table~\ref{tab:multi_session_speedup}, where each row corresponds to one session. The first three columns, $\alpha^*_i$, $\tilde{S}_i(\alpha_i^*, k^*)$, and $\overline{S}_i$, give the optimal ingress, \begin{table}[h]
\centering
\caption{Multi-session performance at optimal ingress $\alpha^*_i$}
\label{tab:multi_session_speedup}
\renewcommand{\arraystretch}{1.3}
\resizebox{\columnwidth}{!}{%
\begin{tabular}{l c c c c c c}
\toprule
Session & $z$ & $\alpha^*_i$ & $\tilde{S}_i(\alpha_i^*, k^*)$ [\%] & $\overline{S}_i$ [\%] & $R_i$ [\%] & $\Delta T_{\text{serv},i}$ \\
\midrule
~$d_{20}$ & 0.0 & 0.442 & 25.2 & $25.4 \pm 0.1$ & 25.6 & 29.2~min \\
~$d_{2}$  & 1.8 & 0.431 & 26.7 & $26.7 \pm 0.1$ & 30.1 & 15.8~min \\
\bottomrule
\end{tabular}%
}
\end{table}the Markov minimum speedup bound, and the sample mean simulated speedup with its 95\% confidence interval over $10^3$ trials. The remaining columns, $R_i$ and $\Delta T_{\text{serv},i}$, give the release time reduction session~$i$ would achieve if run alone at~$\alpha^*_i$, computed against the baseline $E_i + t^{\text{minsep}}_i$ as in Table~\ref{tab:session_speedup}. This framing allows direct comparison with each session's standalone performance in Section~\ref{subsec:session-speedup}.

Co-scheduling raises the minimum session speedup $\tilde{S}_i$ from $0.8\%$ to $25.2\%$ for~$d_{20}$, and from $2.0\%$ to $26.7\%$ for~$d_2$. The co-scheduled PGA fills most of the round-robin cycle on its own, so each session's residual ingress is small relative to the single-session case ($\alpha^*_{\text{eff},20} = 0.047$, $\alpha^*_{\text{eff},2} = 0$). The Markov bound $\tilde{S}_i$ underestimates the simulated speedup $\overline{S}_i$ by at most $0.2$ percentage points, narrower than in the single-session case due to the reduced skip penalty.

\section{Discussion}

The evaluation yields two practical takeaways for entanglement generation scheduling. Firstly, the single-session results reveal that scheduled separation
reduction yields no benefit when packets are produced with
near-certainty within a scheduled block. In this regime, minimum
separation violations become very likely, the resulting skips
accumulate, and execution time grows faster than the separation
reduction recovers. Schedulers can use this criterion to evaluate
whether scheduled separation reduction is indicated for a given
demand. Secondly, co-scheduled sessions provide natural separation between a
session's consumptions, absorbing time that would otherwise be
reserved for the scheduled separation. This enables larger scheduled separation
reductions.

Within both state-of-the-art and projected hardware regimes, the
co-scheduling effect can fully eliminate the scheduled separation:
when the duration of co-scheduled blocks exceeds a session's minimum
separation, no minimum separation violations are possible. The
scheduler then allocates the minimum scheduled time, and the session
correspondingly achieves the maximum reduction in network service
time and maximum execution speedup. Whether this condition holds in practice depends on how the scheduler allocates resources across demands. Network schedulers may partition scheduled resources into non-interacting regions, as Arqon does through its path partition; within such a region a session may have few or no co-scheduled partners, and the single-session optimization determines the largest feasible separation reduction. Conversely, when multiple sessions share the same resources, the duration of intervening co-scheduled blocks is more likely to exceed a session's minimum separation, enabling its scheduled separation to be fully eliminated.

Several directions remain regarding improvements to the analytical
methods and demonstrations of the scheduled separation optimization in
quantum network testbeds. Firstly, within the Arqon architecture,
accepted sessions receive service agreements committing the scheduler
to produce the required packets before expiry with at least the
specified service probability. Monte Carlo simulations confirm that the service commitment is met under
the optimized separation, but the framework does not formally guarantee
it. Establishing this guarantee analytically would convert the empirical
observation into a formal property of the scheduler. Secondly, the Markov lower bound on the expected session execution speedup is conservative. Because packet generation times have
variance, individual sessions on average complete earlier than the
bound predicts. Accounting for this variance would tighten the bound
while preserving its closed form.

For network schedulers that enter service agreements with sessions,
extending the framework to stricter operating regimes raises two
further questions. The evaluation uses a relaxed service probability
to isolate the effect of scheduled separation reduction. Higher service probabilities require the scheduler to allocate
additional blocks to compensate for minimum separation violations.
Jointly optimizing the separation and the block allocation in this
regime remains open.
This joint optimization could also incorporate the session's expiry, guaranteeing the service time meets the deadline.

Although developed within the Arqon architecture, our methods apply to
any architecture in which end nodes share hardware between entanglement
generation and LOCC. Alternative network schedulers or end
node behaviors may yield further gains in scheduling or session
execution efficiency. For example, adaptive scheduling could allocate additional blocks in regimes where minimum separation violations are likely, providing
benefit in cases where uniform reduction yields none. End node policies could also begin consumption mid-block as soon as LOCC operations complete, rather than skipping the block entirely.

\section*{Acknowledgment}
Claude Opus 4.6 co-developed simulations running Alg. 1.

\section{Appendices}
\subsection{Expected Number of Executed PGAs} \label{app:markov}
Let $M_k$ be a Markov chain defined on the state space $\mathbb{N}\times\{G,\overline{G}\}\cup\{(0, \overline{G})\}$, with transition probabilities of Eq.~\ref{eq:transition}, and let us write $M_k = (N_k, S_k)$. We want to calculate $\mathbb{E}[N_k]$. 


\begin{claim}[$k$-step transition matrix of $S_k$]
    Let $P_S$ be the transition matrix for the Markov chain $S_k$. Then the $k$-step transition matrix for $S_k$, $P_S^k$ is given by \begin{equation}
        P_S^k = \begin{bmatrix}
            1 - f_k(pr) & f_k(pr) \\ 1 - g_k(pr) & g_k(pr)
        \end{bmatrix}
    \end{equation} where \begin{align}
        f_k(x) &= x\sum_{i=0}^{k-1}(-x)^i & g_k(x) &= x\sum_{i=0}^{k-2}(-x)^i.
    \end{align}
    \label{claim: reduced chain k-step transition}
\end{claim}

\begin{proof}
    We can write $P_S = P(pr)$, where \begin{equation}
        P(x) = \begin{bmatrix}
            1-x & x \\ 1 & 0
        \end{bmatrix}.
    \end{equation}
    It is then equivalent to show that $$P(x)^k = \begin{bmatrix}
        1 - f_k(x) & f_k(x) \\ 1 - g_k(x) & g_k(x)
    \end{bmatrix}.$$
    
    We proceed by induction.
    For $k=1$, we have $f_1(x) = x$ and $g_1(x) = 0$, so $P(x)^1$ is of the desired form. 

    Assume the results holds for $n = k$. Then \begin{align*}
    P(x)^{k+1} &= P(x)^k P(x)
    = \begin{bmatrix}
        1 - f_{k}(x) & f_k(x) \\ 1 - g_k(x) & g_k(x)
    \end{bmatrix}
    \begin{bmatrix}
        1 - x & x \\ 1 & 0
    \end{bmatrix} \\[6pt]
    &= \begin{bmatrix}
        1 - x(1-f_k(x)) & x(1-f_k(x)) \\
        1 - x(1-g_k(x)) & x(1-g_k(x))
    \end{bmatrix}
\end{align*}
  By direct calculation one can show $x(1-f_k(x)) = f_{k+1}(x)$, $x(1-g_k(x)) = g_{k+1}(x)$ and the result follows. 
\end{proof}

\begin{claim}
    $\mathbb{E}[N_k] = \mathbb{E}[N_{k-1}] + 1 - f_{k-1}(pr)$
    \label{claim: recurrence relation}
\end{claim}
\begin{proof}\footnotesize
    \begin{align}
    \mathbb{E}[N_k] &= \sum_{n} n\,\mathbb{P}[N_k = n] \notag\\
    &= \sum_{n,m} n\,\mathbb{P}[N_k = n \mid N_{k-1} = m]\,\mathbb{P}[N_{k-1} = m] \notag\\
    &= \sum_{n,m} n\Big(
        \mathbb{P}[N_k = n \mid N_{k-1} = m,\, S_{k-1} = G]\,\mathbb{P}[S_{k-1} = G] \notag\\
    &\qquad\quad + \mathbb{P}[N_k = n \mid N_{k-1} = m,\, S_{k-1} = \bar{G}]\,\mathbb{P}[S_{k-1} = \bar{G}]
    \Big) \notag\\
    &\qquad\quad \times \mathbb{P}[N_{k-1} = m]. \label{eq: E[N_k] pt1}
\end{align}\normalsize

    Now, we have that \small \begin{align*} 
        \mathbb{P}[N_k =n | N_{k-1} = m, S_{k-1}=\overline{G}] &= \left\lbrace\begin{matrix}
            1 & n = m + 1\\ 0 & \text{else}
        \end{matrix}\right. \\
        \mathbb{P}[N_k =n | N_{k-1} = m, S_{k-1}=G] &= \left\lbrace\begin{matrix}
            1 & n = m \\ 0 & \text{else}
        \end{matrix}\right.
    \end{align*} and so from \eqref{eq: E[N_k] pt1}, we get \begin{align*}
    \mathbb{E}[N_k] &= \sum_{m} \Big(
        m\,\mathbb{P}[S_{k-1} = G]
        + (m+1)\,\mathbb{P}[S_{k-1} = \bar{G}]
    \Big) \\
    &\qquad\quad \times \mathbb{P}[N_{k-1} = m].
\end{align*}\normalsize
    From the earlier claim, we have that $\mathbb{P}[S_{k-1} = {G}] = 1 - \mathbb{P}[S_{k-1} = \overline{G}] = f_{k-1}(pr)$. Therefore, we have \footnotesize \begin{align*}
        \mathbb{E}[N_k] &= \sum_m\big(mf_{k-1}(pr) + (1+m)(1-f_{k-1}(pr))\big)\mathbb{P}[N_{k-1}=m]\\
        &= \sum_m m\mathbb{P}[N_{k-1}=m] + (1 - f_{k-1}(pr))\sum_m\mathbb{P}[N_{k-1} = m] \\
        &= \mathbb{E}[N_{k-1}] + 1 - f_{k-1}(pr)
    \end{align*}\normalsize
\end{proof}

\begin{claim}
    $\mathbb{E}[N_k] = \sum_{i=0}^{k-1}(k-i)(-pr)^i$
\end{claim}
\begin{proof}
    We proceed by induction. By inspection of the corresponding Markov chain in Figure~\ref{fig:main chain}, we can see that given we start in state $(0, \overline{G})$, we must have $\mathbb{E}[N_1] = 1$. From the desired form, we have $\mathbb{E}[N_{1}] = \sum_{i=0}^0(k-i)(-pr)^i = 1$. So the base case holds. 

    Let us now assume that the assumption holds for $n=k$. Then, using the recurrence relation from Claim~\ref{claim: recurrence relation}, we have that
    $\mathbb{E}[N_{k+1}] = \mathbb{E}[N_{k}] + 1 - f_{k+1}$ and the result follows by direct calculation.
\end{proof}

\bibliography{ref}

@article{beauchamp2025modular,
  title={A modular quantum network architecture for integrating network scheduling with local program execution},
  author={Beauchamp, Thomas R and Jirovsk{\'a}, Hana and Gauthier, Scarlett and Wehner, Stephanie},
  journal={arXiv preprint arXiv:2503.12582},
  year={2025}
}

@misc{arqon_arxiv,
      title={Arqon: A suite of control applications enabling a reliable quantum network}, 
      author={Scarlett Gauthier and Thomas R. Beauchamp and Stephanie Wehner},
      year={2026},
      eprint={2604.08692},
      archivePrefix={arXiv},
      primaryClass={quant-ph},
      url={https://arxiv.org/abs/2604.08692}, 
}

@misc{minsep_relaxation_2026,
  author       = {Jake Smith},
  title        = {{Data-set and Software Tools: Tools for Reducing Service Time in Near-Term Quantum Networks}},
  year         = {2026},
  howpublished = {\url{https://gitlab.tudelft.nl/wehner-research/minsep_relaxation}},
}

@inproceedings{shi2020concurrent,
  title={Concurrent entanglement routing for quantum networks: Model and designs},
  author={Shi, Shouqian and Qian, Chen},
  booktitle={Proceedings of the Annual conference of the ACM Special Interest Group on Data Communication on the applications, technologies, architectures, and protocols for computer communication},
  pages={62--75},
  year={2020}
}

@article{canteri2024photon,
  title={A photon-interfaced ten qubit quantum network node},
  author={Canteri, M and Koong, ZX and Bate, J and Winkler, A and Krutyanskiy, V and Lanyon, BP},
  journal={arXiv preprint arXiv:2406.09480},
  year={2024}
}

@article{vasantam2021stability,
  title={Stability analysis of a quantum network with max-weight scheduling},
  author={Vasantam, Thirupathaiah and Towsley, Don},
  journal={arXiv preprint arXiv:2106.00831},
  year={2021}
}

@article{krutyanskiy2023telecom,
  title={Telecom-wavelength quantum repeater node based on a trapped-ion processor},
  author={Krutyanskiy, Victor and Canteri, Marco and Meraner, Martin and Bate, James and Krcmarsky, Vojtech and Schupp, Josef and Sangouard, Nicolas and Lanyon, Ben P},
  journal={Physical Review Letters},
  volume={130},
  number={21},
  pages={213601},
  year={2023},
  publisher={APS}
}

@inproceedings{pouryousef2023quantum,
  title={A quantum overlay network for efficient entanglement distribution},
  author={Pouryousef, Shahrooz and Panigrahy, Nitish K and Towsley, Don},
  booktitle={IEEE INFOCOM 2023-IEEE Conference on Computer Communications},
  pages={1--10},
  year={2023},
  organization={IEEE}
}

@article{dai2020optimal,
  title={Optimal remote entanglement distribution},
  author={Dai, Wenhan and Peng, Tianyi and Win, Moe Z},
  journal={IEEE Journal on Selected Areas in Communications},
  volume={38},
  number={3},
  pages={540--556},
  year={2020},
  publisher={IEEE}
}

@article{cicconetti2021request,
  title={Request scheduling in quantum networks},
  author={Cicconetti, Claudio and Conti, Marco and Passarella, Andrea},
  journal={IEEE Transactions on Quantum Engineering},
  volume={2},
  pages={2--17},
  year={2021},
  publisher={IEEE}
}

@article{li2021effective,
  title={Effective routing design for remote entanglement generation on quantum networks},
  author={Li, Changhao and Li, Tianyi and Liu, Yi-Xiang and Cappellaro, Paola},
  journal={npj Quantum Information},
  volume={7},
  number={1},
  pages={10},
  year={2021},
  publisher={Nature Publishing Group UK London}
}

@article{pant2019routing,
  title={Routing entanglement in the quantum internet},
  author={Pant, Mihir and Krovi, Hari and Towsley, Don and Tassiulas, Leandros and Jiang, Liang and Basu, Prithwish and Englund, Dirk and Guha, Saikat},
  journal={npj Quantum Information},
  volume={5},
  number={1},
  pages={25},
  year={2019},
  publisher={Nature Publishing Group UK London}
}

@inproceedings{zhao2021redundant,
  title={Redundant entanglement provisioning and selection for throughput maximization in quantum networks},
  author={Zhao, Yangming and Qiao, Chunming},
  booktitle={IEEE INFOCOM 2021-IEEE Conference on Computer Communications},
  pages={1--10},
  year={2021},
  organization={IEEE}
}

@article{van2025qoala,
  title={Qoala: an application execution environment for quantum internet nodes},
  author={van der Vecht, Bart and Y{\"u}cel, Atak Talay and Jirovsk{\'a}, Hana and Wehner, Stephanie},
  journal={arXiv preprint arXiv:2502.17296},
  year={2025}
}

@article{pompili2022experimental,
  title={Experimental demonstration of entanglement delivery using a quantum network stack},
  author={Pompili, Matteo and Delle Donne, Carlo and te Raa, Ingmar and van der Vecht, Bart and Skrzypczyk, Matthew and Ferreira, Guilherme and de Kluijver, Lisa and Stolk, Arian J and Hermans, Sophie LN and Pawe{\l}czak, Przemys{\l}aw and others},
  journal={npj Quantum Information},
  volume={8},
  number={1},
  pages={121},
  year={2022},
  publisher={Nature Publishing Group UK London}
}

@article{pirker2019quantum,
  title={A quantum network stack and protocols for reliable entanglement-based networks},
  author={Pirker, Alexander and D{\"u}r, Wolfgang},
  journal={New Journal of Physics},
  volume={21},
  number={3},
  pages={033003},
  year={2019},
  publisher={IOP Publishing}
}

@inproceedings{gu2023esdi,
  title={Esdi: Entanglement scheduling and distribution in the quantum internet},
  author={Gu, Huayue and Yu, Ruozhou and Li, Zhouyu and Wang, Xiaojian and Zhou, Fangtong},
  booktitle={2023 32nd International Conference on Computer Communications and Networks (ICCCN)},
  pages={1--10},
  year={2023},
  organization={IEEE}
}

@article{stolk2024metropolitan,
  title={Metropolitan-scale heralded entanglement of solid-state qubits},
  author={Stolk, Arian J and van der Enden, Kian L and Slater, Marie-Christine and te Raa-Derckx, Ingmar and Botma, Pieter and Van Rantwijk, Joris and Biemond, JJ Benjamin and Hagen, Ronald AJ and Herfst, Rodolf W and Koek, Wouter D and others},
  journal={Science advances},
  volume={10},
  number={44},
  pages={eadp6442},
  year={2024},
  publisher={American Association for the Advancement of Science}
}

@inproceedings{broadbent2009universal,
  title={Universal blind quantum computation},
  author={Broadbent, Anne and Fitzsimons, Joseph and Kashefi, Elham},
  booktitle={2009 50th annual IEEE symposium on foundations of computer science},
  pages={517--526},
  year={2009},
  organization={IEEE}
}

@article{bova2021commercial,
  title={Commercial applications of quantum computing},
  author={Bova, Francesco and Goldfarb, Avi and Melko, Roger G},
  journal={EPJ quantum technology},
  volume={8},
  number={1},
  pages={2},
  year={2021},
  publisher={Springer}
}

@article{wehner2018quantum,
  title={Quantum internet: A vision for the road ahead},
  author={Wehner, Stephanie and Elkouss, David and Hanson, Ronald},
  journal={Science},
  volume={362},
  number={6412},
  pages={eaam9288},
  year={2018},
  publisher={American Association for the Advancement of Science}
}

@article{zaiser2016enhancing,
  title={Enhancing quantum sensing sensitivity by a quantum memory},
  author={Zaiser, Sebastian and Rendler, Torsten and Jakobi, Ingmar and Wolf, Thomas and Lee, Sang-Yun and Wagner, Samuel and Bergholm, Ville and Schulte-Herbr{\"u}ggen, Thomas and Neumann, Philipp and Wrachtrup, J{\"o}rg},
  journal={Nature communications},
  volume={7},
  number={1},
  pages={12279},
  year={2016},
  publisher={Nature Publishing Group UK London}
}

@article{giovannetti2001quantum,
  title={Quantum-enhanced positioning and clock synchronization},
  author={Giovannetti, Vittorio and Lloyd, Seth and Maccone, Lorenzo},
  journal={Nature},
  volume={412},
  number={6845},
  pages={417--419},
  year={2001},
  publisher={Nature Publishing Group UK London}
}

@article{delle2025operating,
  title={An operating system for executing applications on quantum network nodes},
  author={Delle Donne, Carlo and Iuliano, Mariagrazia and van der Vecht, Bart and Ferreira, Guilherme Maciel and Jirovsk{\'a}, Hana and van der Steenhoven, Thom JW and Dahlberg, Axel and Skrzypczyk, Matthew and Fioretto, Dario and Teller, Markus and others},
  journal={Nature},
  volume={639},
  number={8054},
  pages={321--328},
  year={2025},
  publisher={Nature Publishing Group UK London}
}

@misc{arqon_sim,
  doi = {10.4121/97e0f340-6cc2-4908-9670-8308fed3e2ec.v1},
  url = {},
  author = {Gauthier, Scarlett and Beauchamp, Thomas R. and Wehner, Stephanie},
  title = {Code and Data for Arqon: A Suite of Control Applications Enabling Reliable Quantum Networks},
  publisher = {4TU.ResearchData},
  year = {2026},
  copyright = {CC BY-NC 4.0},
}

@article{hucul2015modular,
  title={Modular entanglement of atomic qubits using photons and phonons},
  author={Hucul, David and Inlek, Ismail V and Vittorini, Grahame and Crocker, Clayton and Debnath, Shantanu and Clark, Susan M and Monroe, Christopher},
  journal={Nature Physics},
  volume={11},
  number={1},
  pages={37--42},
  year={2015},
  publisher={Nature Publishing Group UK London}
}

@article{liu2026long,
  title={Long-lived remote ion-ion entanglement for scalable quantum repeaters},
  author={Liu, Wen-Zhao and Zhou, Ya-Bin and Chen, Jiu-Peng and Wang, Bin and Teng, Ao and Han, Xiao-Wen and Liu, Guang-Cheng and Zhang, Zhi-Jiong and Yang, Yi and Liu, Feng-Guang and others},
  journal={Nature},
  pages={1--3},
  year={2026},
  publisher={Nature Publishing Group UK London}
}

@INPROCEEDINGS {two_level_control,
author = { Yu, Se-young and Perego, Elia and Phillips, Justin and Cheah, You-Wei and Umesh, Prathwiraj and Gao, Guangqi and Liu, Jiarui and Kissel, Ezra and Bregar, Michael and Sun, Ke and Wu, Qiming and Valivarthi, Raju and Saglamyurek, Erhan and Wu, Wenji and Spiropulu, Maria and Haffner, Hartmut and Monga, Inder },
booktitle = { 2025 IEEE International Conference on Quantum Computing and Engineering (QCE) },
title = {{ A Two-Level Control Framework for Quantum Networks }},
year = {2025},
volume = {},
ISSN = {},
pages = {1302-1311},
doi = {10.1109/QCE65121.2025.00145},
url = {https://doi.ieeecomputersociety.org/10.1109/QCE65121.2025.00145},
publisher = {IEEE Computer Society},
address = {Los Alamitos, CA, USA},
month =sep}

@article{islam2025experimental,
  title={Experimental Demonstration of Software-Orchestrated Quantum Network Applications over a Campus-Scale Testbed},
  author={Islam, Md Shariful and Chung, Joaquin and Eastman, Ely Marcus and Hayek, Robert J and Kumar, Prem and Kettimuthu, Rajkumar},
  journal={arXiv preprint arXiv:2511.01247},
  year={2025}
}

@inproceedings{harrison2010introduction,
  title={Introduction to {M}onte {C}arlo simulation},
  author={R. L. Harrison},
  booktitle={AIP Conf. Proc.},
  volume={1204},
  number={1},
  pages={17ß},
  year={2010},
  organization={Amer. Inst. Phys.}
}

@article{kapoor2025public,
  title={Public quantum network: The first node},
  author={Kapoor, K and Hoseini, S and Choi, J and Nussbaum, BE and Zhang, Y and Shetty, K and Skaar, C and Ward, M and Wilson, L and Shinbrough, K and others},
  journal={Applied Physics Letters},
  volume={126},
  number={5},
  year={2025},
  publisher={AIP Publishing}
}

\end{document}